\documentclass[%
 reprint,
 amsmath,amssymb,
 aps,
]{revtex4-2}

\usepackage{graphicx}% Include figure files
\usepackage{dcolumn}% Align table columns on decimal point
\usepackage{bm}% bold math
\usepackage{amsthm}% theorem environments
\usepackage{mathtools}
\usepackage{dsfont}
\usepackage{hyperref}% hypertext capabilities

\newtheorem{lemma}{Lemma}
\newtheorem{proposition}{Proposition}
\newtheorem{corollary}{Corollary}
\theoremstyle{definition}

\newtheorem{construction}{Construction}
\newtheorem{remark}{Remark}

\newcommand{\id}{\mathds{1}}
\newcommand{\Hil}{\mathcal{H}}
\newcommand{\Lin}{\mathcal{L}}
\newcommand{\C}{\mathbb{C}}

\newcommand{\N}{\mathbb{N}}
\newcommand{\Z}{\mathbb{Z}}
\newcommand{\Tr}{\operatorname{tr}}
\newcommand{\dg}{^{\dagger}}
\newcommand{\inner}[2]{\langle #1,\,#2\rangle}
\newcommand{\chconj}{\chi_{\mathrm{conj}}}
\newcommand{\Mconj}{\mathcal{M}}
\newcommand{\Bell}{\Phi^{+}}
\newcommand{\sigmax}{\sigma_x}
\newcommand{\sigmay}{\sigma_y}
\newcommand{\sigmaz}{\sigma_z}

\begin{document}

\preprint{APS/123-QED}

\title{Quantum Realizability of Probabilistic Theories Stable under Teleportation}

\author{Miguel A.\ A.\ Lisboa}
 \email{miguel.lisboa@venturus.org.br}
 \affiliation{%
  Universidade Federal da Para\'iba, Centro de Inform\'atica,
  Jo\~ao Pessoa, Brazil
 }%
 \affiliation{%
  Venturus, Centro de Excel\^encia em Computa\c{c}\~ao Qu\^antica,
  Campinas, Brazil
 }%

\date{\today}

\begin{abstract}
The classification of general probabilistic theories (GPTs) whose CHSH value is stable under arbitrary rounds of teleportation and entanglement swapping was obtained by Dmello and Gross~\cite{DmelloGross2026} and consists of seven families, indexed by characters of the Klein four-group $K_4$, the cyclic group $\Z_4$, and the dihedral group $D_4$. The question of which of these families admits a realization within standard quantum mechanics was left open. In this work we resolve this question completely. Using elementary representation theory, we prove that exactly two families are quantum-realizable, namely $\chi^{K_4}_{1234}$ and $\chi^{D_4}_{125}$, while the remaining five admit no quantum realization.
\end{abstract}

\maketitle

%-----------------------------------------------------------------------
\section{\label{sec:intro}Introduction}
%-----------------------------------------------------------------------

The framework of general probabilistic theories (GPTs) provides an operational description of physics that subsumes both classical and quantum theory and allows one to study the structural origin of quantum features such as entanglement, no-cloning~\cite{Wootters1982}, and information-theoretic protocols~\cite{Barrett2007,Chiribella2011,Chiribella2010,Plavala2023,Hardy2001,Masanes2011,JanottaHinrichsen2014,BarrnumWilce2016,Spekkens2007}. A central problem in this program is to identify operational principles that single out quantum theory among the broader landscape of GPTs. Quantum teleportation~\cite{Bennett1993,Nielsen2000} and entanglement swapping~\cite{Zukowski1993,Pan1998,Bose1998} provide a particularly natural setting for this question: requiring that the CHSH value~\cite{CHSH1969,Bell1964,Tsirelson1980,Brunner2014} of a bipartite system remain invariant under arbitrary rounds of teleportation and entanglement swapping is a highly non-trivial structural constraint, and the original work of Weilenmann and Colbeck~\cite{WeilenmannColbeck2020a,WeilenmannColbeck2020b} conjectured that such stability could potentially serve as an operational characterization of quantum mechanics.

In Ref.~\cite{DmelloLighthart2024}, the authors gave a first counter-example to this conjecture by exhibiting a post-quantum GPT~\cite{Popescu1994}, called oblate stabilizer theory, which sustains the algebraic maximum CHSH value of $4$ under arbitrary iterated entanglement swapping. Building on these results, Dmello and Gross~\cite{DmelloGross2026} subsequently classified all GPTs whose CHSH value is stable in this sense. Under mild technical hypotheses they showed that any such effective GPT must belong to one of seven inequivalent families, characterized by the irreducible-decomposition data of a representation of the correction group acting on the underlying state space. These families are labeled by characters of three finite groups: the cyclic group $\Z_4$, the Klein four-group $K_4$, and the dihedral group $D_4$ of order $8$. The seven characters are
\begin{equation}
\label{eq:seven-families}
\chi^{K_4}_{1234},\;\chi^{\Z_4}_{1234},\;\chi^{D_4}_{125},\;\chi^{D_4}_{135},\;\chi^{D_4}_{145},\;\chi^{D_4}_{12345},\;\chi^{D_4}_{123452}.
\end{equation}
The first two are the regular characters of $K_4$ and $\Z_4$ respectively; the remaining five are characters of $D_4$ specified by the multiplicities of its five irreducible representations.

The question, posed in Ref.~\cite{DmelloGross2026} and resolved here, is the following. Among these seven abstract symmetry types, which can be physically realized within quantum theory?

We make precise what is meant by ``quantum-realizable.'' A quantum realization of a family in Eq.~\eqref{eq:seven-families} consists of a finite-dimensional Hilbert space $\Hil$ and a (possibly projective) unitary representation $U:G\to\mathrm{U}(\Hil)$ of the relevant correction group such that the conjugation character $\chconj$ of the action $\Mconj_g(X)=U_gXU_g\dg$ on the full algebra $\Lin(\Hil)$ equals the prescribed family. The problem is therefore one of projective-unitary realizability, and the paper proves exactly the classification of which families are projective-unitarily realizable in this sense.

The main result of this paper is a complete classification: exactly two of the seven families admit a quantum realization, namely $\chi^{K_4}_{1234}$ and $\chi^{D_4}_{125}$, while the remaining five do not. The proof is uniform and entirely representation-theoretic. The fundamental input is the identity $\chconj=|\chi_U|^2$ relating the conjugation character to the squared modulus of the underlying unitary character. Combined with the constraint that the trivial character occurs with multiplicity one in $\chconj$~\cite[Thm.~11]{DmelloGross2026}, this forces $\chi_U$ to be irreducible, and the analysis then reduces to checking a small number of arithmetic identities on character tables. We complement the impossibility results with explicit quantum constructions of the two realizable families, expressed in operational language as the standard Bell-basis teleportation protocol~\cite{Bennett1993} and a POVM-based entanglement-swapping protocol~\cite{DmelloLighthart2024}.

%-----------------------------------------------------------------------
\section{\label{sec:prelim}Notation and Background}
%-----------------------------------------------------------------------

Throughout, $G$ denotes a finite group and $\Hil$ a finite-dimensional complex Hilbert space. We write $\Lin(\Hil)$ for the algebra of linear operators on $\Hil$, equipped with the Hilbert--Schmidt inner product
\begin{equation}
\langle X,Y\rangle_{\mathrm{HS}}:=\Tr(X\dg Y),\qquad X,Y\in\Lin(\Hil).
\end{equation}
If $\dim\Hil=n$, then $\Lin(\Hil)\cong M_n(\C)$ as a vector space, and $\dim\Lin(\Hil)=n^2$.

A unitary representation of $G$ on $\Hil$ is a group homomorphism $U:G\to\mathrm{U}(\Hil)$, $g\mapsto U_g$, where $\mathrm{U}(\Hil)$ denotes the group of unitary operators on $\Hil$. The character of $U$ is the function
\begin{equation}
\chi_U:G\to\C,\qquad\chi_U(g):=\Tr(U_g).
\end{equation}
Characters are class functions, meaning $\chi_U(hgh^{-1})=\chi_U(g)$ for all $g,h\in G$. The standard inner product on class functions of $G$ is
\begin{equation}
\inner{\chi}{\chi'}:=\frac{1}{|G|}\sum_{g\in G}\chi(g)^{*}\chi'(g).
\end{equation}
The irreducible characters $\chi_1,\dots,\chi_k$ of $G$ form an orthonormal basis for the space of class functions with respect to this inner product, by the Schur orthogonality relations~\cite{Serre1977,Isaacs1976,JamesLiebeck2001}:
\begin{equation}
\inner{\chi_i}{\chi_j}=\delta_{ij}.
\end{equation}
Every character decomposes as $\chi_U=\sum_{i}n_i\chi_i$ with $n_i=\inner{\chi_i}{\chi_U}\in\N_0$, and the representation is irreducible if and only if $\inner{\chi_U}{\chi_U}=1$.

\subsection{The conjugation representation}

Given a unitary representation $U$ of $G$ on $\Hil$, the \emph{conjugation representation} is the representation of $G$ on $\Lin(\Hil)$ defined by
\begin{equation}
\Mconj_g(X):=U_gXU_g\dg,\qquad X\in\Lin(\Hil),\;g\in G.
\end{equation}
This is indeed a representation: $\Mconj_e(X)=U_eXU_e\dg=\id X\id=X$, and
\begin{equation}
\Mconj_g\circ\Mconj_h(X)=U_gU_hXU_h\dg U_g\dg=U_{gh}XU_{gh}\dg=\Mconj_{gh}(X)
\end{equation}
by the homomorphism property. Each $\Mconj_g$ is unitary with respect to the Hilbert--Schmidt inner product:
\begin{align}
\langle\Mconj_g(X),\Mconj_g(Y)\rangle_{\mathrm{HS}}
  &=\Tr\!\bigl((U_gXU_g\dg)\dg\,U_gYU_g\dg\bigr)\notag\\
  &=\Tr\!\bigl(U_gX\dg U_g\dg U_gYU_g\dg\bigr)\notag\\
  &=\Tr\!\bigl(U_gX\dg YU_g\dg\bigr)\notag\\
  &=\Tr(X\dg Y)
   =\langle X,Y\rangle_{\mathrm{HS}},
\end{align}
where we used $U_g\dg U_g=\id$ and the cyclicity of the trace. The character of $\Mconj$ is denoted $\chconj$.

\subsection{Projective representations and the Schur cover}

A projective unitary representation of $G$ on $\Hil$ is a map $U:G\to\mathrm{U}(\Hil)$ satisfying $U_gU_h=\omega(g,h)U_{gh}$ for some scalar $\omega(g,h)\in\mathrm{U}(1)$. The function $\omega:G\times G\to\mathrm{U}(1)$ is a $2$-cocycle, called the multiplier of $U$. Projective representations with the same multiplier form a projective class; two multipliers are equivalent if they differ by a coboundary. Equivalence classes of multipliers are classified by the second cohomology group $H^2(G,\mathrm{U}(1))$.

Note that the conjugation action $X\mapsto U_gXU_g\dg$ is unchanged if $U_g$ is replaced by $e^{i\theta}U_g$, so the conjugation representation $\Mconj$ is well-defined for projective representations.

For $G=D_4$, the cohomology group is $H^2(D_4,\mathrm{U}(1))\cong\Z_2$~\cite{Karpilovsky1985,Brown1982}, so there are exactly two projective classes. In the trivial class, $U$ is an ordinary linear representation. In the non-trivial class, $U$ lifts to an ordinary representation of the Schur cover of $D_4$. The Schur cover is a central extension
\begin{equation}
1\to\Z_2\to\widetilde{G}\to D_4\to 1
\end{equation}
in which the central $\Z_2$ generator $z$ acts on $\Hil$ by $U_z=-\id$. We may take $\widetilde{G}$ to be the dihedral group $D_8$ of order $16$.

\subsection{Groups and their character tables}

We fix presentations
\begin{align}
K_4&=\{e,a,b,ab\},\quad a^2=b^2=(ab)^2=e,\\
\Z_4&=\langle t\mid t^4=e\rangle,\\
D_4&=\langle r,s\mid r^4=s^2=e,\;srs^{-1}=r^{-1}\rangle,\\
D_8&=\langle\zeta,\eta\mid\zeta^8=\eta^2=e,\;\eta\zeta^k\eta^{-1}=\zeta^{-k}\rangle.
\end{align}

The conjugacy classes of $D_4$, their representatives, and class sizes are listed in Table~\ref{tab:D4classes}.%
\begin{table}[b]
\caption{\label{tab:D4classes}Conjugacy classes of $D_4$.}
\begin{ruledtabular}
\begin{tabular}{ccc}
Class & Representative & Size\\
\colrule
$K_1$     & $e$   & $1$\\
$K_{r^2}$ & $r^2$ & $1$\\
$K_r$     & $r$   & $2$\\
$K_s$     & $s$   & $2$\\
$K_{rs}$  & $rs$  & $2$\\
\end{tabular}
\end{ruledtabular}
\end{table}

These classes are obtained directly from the defining relation, which gives $sr^k=r^{-k}s$ and equivalently $rs=sr^{-1}$. Since $srs^{-1}=r^{-1}=r^3$, we have $K_r=\{r,r^3\}$. For the reflections, $rsr^{-1}=(sr^{-1})r^{-1}=sr^{-2}=r^2 s$, so $K_s=\{s,r^2s\}$. Similarly, $r(rs)r^{-1}=r^2 s r^{-1}=r^2(sr^{-1})=r^2\cdot rs=r^3 s$, so $K_{rs}=\{rs,r^3s\}$. The five irreducible characters of $D_4$ at the representatives $(e,r,r^2,s,rs)$ are
\begin{align}
\chi_1&=(1,1,1,1,1),\\
\chi_2&=(1,1,1,-1,-1),\\
\chi_3&=(1,-1,1,1,-1),\\
\chi_4&=(1,-1,1,-1,1),\\
\chi_5&=(2,0,-2,0,0).
\end{align}
Here $\chi_1,\dots,\chi_4$ are the four one-dimensional characters, and $\chi_5$ is the unique two-dimensional irreducible. The dimensions satisfy $1^2+1^2+1^2+1^2+2^2=8=|D_4|$, which confirms this is the complete list. The regular character of $D_4$ has value $|D_4|=8$ at $e$ and $0$ elsewhere; it decomposes as $\sum_i\chi_i(e)\chi_i=\chi_1+\chi_2+\chi_3+\chi_4+2\chi_5$.

For $K_4=\Z_2\times\Z_2$, all four irreducible characters are one-dimensional. Writing the elements as $\{e,a,b,ab\}$ with $a^2=b^2=(ab)^2=e$, the four characters are
\begin{align}
\chi^{K_4}_1&=(1,1,1,1),\\
\chi^{K_4}_2&=(1,-1,1,-1),\\
\chi^{K_4}_3&=(1,1,-1,-1),\\
\chi^{K_4}_4&=(1,-1,-1,1),
\end{align}
at representatives $(e,a,b,ab)$. The regular character of $K_4$ is $\chi^{K_4}_{1234}:=\chi^{K_4}_1+\chi^{K_4}_2+\chi^{K_4}_3+\chi^{K_4}_4$, which takes value $4$ at $e$ and $0$ elsewhere.

For $\Z_4=\{e,t,t^2,t^3\}$, the four irreducible characters are $\chi^{\Z_4}_k(t^j)=i^{kj}$ for $k,j\in\{0,1,2,3\}$. The regular character of $\Z_4$ is $\chi^{\Z_4}_{1234}:=\sum_{k=0}^3\chi^{\Z_4}_k$, taking value $4$ at $e$ and $0$ elsewhere.

The conjugacy classes of $D_8$, at representatives $(e,\zeta,\zeta^2,\zeta^3,\zeta^4,\eta,\zeta\eta)$, have sizes $(1,2,2,2,1,4,4)$. The two-dimensional irreducible characters of $D_8$ that do not factor through the quotient $D_8/\langle\zeta^4\rangle\cong D_4$ are
\begin{align}
\chi_{E_1}&=(2,\sqrt{2},0,-\sqrt{2},-2,0,0),\\
\chi_{E_3}&=(2,-\sqrt{2},0,\sqrt{2},-2,0,0),
\end{align}
and a third two-dimensional character
\begin{equation}
\chi_{E_2}=(2,0,-2,0,2,0,0)
\end{equation}
does factor through the quotient and therefore belongs to the trivial projective class of $D_4$. The characters $\chi_{E_1}$ and $\chi_{E_3}$ correspond to irreducible projective representations of $D_4$ in the non-trivial cohomology class, because $\chi_{E_1}(\zeta^4)=\chi_{E_3}(\zeta^4)=-2=-\chi_{E_1}(e)$, reflecting the fact that the central element $\zeta^4$ acts as $-\id$.

\subsection{The families and their dimensions}

The seven families in Eq.~\eqref{eq:seven-families} are characters of the conjugation representation, hence they must equal $\chconj$ for some unitary representation $U$. Their dimensions (values at the group identity) are recorded in Table~\ref{tab:dimensions}.%
\begin{table}[b]
\caption{\label{tab:dimensions}Dimensions of the seven families. The subscript lists the irreducible characters included; the trailing $2$ in $\chi^{D_4}_{123452}$ means $\chi_5$ appears with multiplicity $2$.}
\begin{ruledtabular}
\begin{tabular}{lc}
Family & $\dim = \chconj(e)$\\
\colrule
$\chi^{K_4}_{1234}$   & $4$ \\
$\chi^{\Z_4}_{1234}$  & $4$ \\
$\chi^{D_4}_{125}$    & $1+1+2=4$ \\
$\chi^{D_4}_{135}$    & $1+1+2=4$ \\
$\chi^{D_4}_{145}$    & $1+1+2=4$ \\
$\chi^{D_4}_{12345}$  & $1+1+1+1+2=6$ \\
$\chi^{D_4}_{123452}$ & $1+1+1+1+2{\cdot}2=8$ \\
\end{tabular}
\end{ruledtabular}
\end{table}

%-----------------------------------------------------------------------
\section{\label{sec:conj-character}The Conjugation Character and Irreducibility}
%-----------------------------------------------------------------------

We establish the fundamental identity relating $\chconj$ to $\chi_U$, and derive the irreducibility constraint from it.

\begin{lemma}\label{lem:conj}
For every $g\in G$ and every unitary representation $U$ of $G$ on $\Hil$,
\begin{equation}
\chconj(g)=|\chi_U(g)|^2=\chi_U(g)^*\chi_U(g).
\end{equation}
\end{lemma}

\begin{proof}
Let $d=\dim\Hil$ and pick any orthonormal basis $\{|\psi_i\rangle\}_{i=1}^{d}$ of $\Hil$. Define the rank-one operators $E_{ij}:=|\psi_i\rangle\langle\psi_j|$ for $i,j\in\{1,\dots,d\}$. Then $\{E_{ij}\}$ is an orthonormal basis of $\Lin(\Hil)$ with respect to the Hilbert--Schmidt inner product:
\begin{align}
\langle E_{ij},E_{kl}\rangle_{\mathrm{HS}}
  &=\Tr(E_{ij}\dg E_{kl})\notag\\
  &=\Tr(|\psi_j\rangle\langle\psi_i|\psi_k\rangle\langle\psi_l|)\notag\\
  &=\delta_{ik}\,\Tr(|\psi_j\rangle\langle\psi_l|)=\delta_{ik}\delta_{jl}.
\end{align}
Since $\dim\Lin(\Hil)=d^2$ and we have $d^2$ orthonormal vectors, they span $\Lin(\Hil)$.

The character of $\Mconj$ is, by definition, the trace of $\Mconj_g$ as a linear map on $\Lin(\Hil)$ using any orthonormal basis. Using $\{E_{ij}\}$:
\begin{align}
\chconj(g)
  &=\sum_{i,j=1}^{d}\langle E_{ij},\Mconj_g(E_{ij})\rangle_{\mathrm{HS}}\notag\\
  &=\sum_{i,j=1}^{d}\Tr\!\bigl(E_{ij}\dg\,(U_g E_{ij}U_g\dg)\bigr)\notag\\
  &=\sum_{i,j=1}^{d}\Tr\!\bigl(|\psi_j\rangle\langle\psi_i|\,
       U_g\,|\psi_i\rangle\langle\psi_j|\,U_g\dg\bigr).
\end{align}
By cyclicity of the trace,
\begin{equation}
\chconj(g)
  =\sum_{i,j}\langle\psi_i|U_g|\psi_i\rangle\,\langle\psi_j|U_g\dg|\psi_j\rangle,
\end{equation}
and the double sum factors:
\begin{align}
\chconj(g)
  &=\Bigl(\sum_{i=1}^{d}\langle\psi_i|U_g|\psi_i\rangle\Bigr)
     \Bigl(\sum_{j=1}^{d}\langle\psi_j|U_g\dg|\psi_j\rangle\Bigr)\notag\\
  &=\Tr(U_g)\cdot\overline{\Tr(U_g)}\notag\\
  &=\chi_U(g)\cdot\overline{\chi_U(g)}=|\chi_U(g)|^2,
\end{align}
where we used $\Tr(U_g\dg)=\overline{\Tr(U_g)}$ (since $U_g\dg$ has matrix entries that are complex conjugates of the transpose of $U_g$, and the trace is invariant under transpose).
\end{proof}

\begin{remark}
Lemma~\ref{lem:conj} is equivalent to the representation-theoretic isomorphism $\Mconj\cong U\otimes\overline{U}$, where $\overline{U}$ is the complex-conjugate representation defined by $\overline{U}_g:=\overline{U_g}$ (complex conjugate of the matrix entries in any fixed basis). The character of $U\otimes\overline{U}$ is
\begin{equation}
\chi_{U\otimes\overline{U}}(g)=\chi_U(g)\cdot\overline{\chi_U(g)}=|\chi_U(g)|^2.
\end{equation}
The isomorphism $\Mconj\cong U\otimes\overline{U}$ is established by the map $|\psi_i\rangle\langle\psi_j|\mapsto|\psi_i\rangle\otimes\langle\psi_j|$, which intertwines the two actions.
\end{remark}

The next step is to translate the structural constraint imposed by Ref.~\cite{DmelloGross2026} into a constraint on $\chi_U$. By Theorem~11 of Ref.~\cite{DmelloGross2026}, every family in Eq.~\eqref{eq:seven-families} has the property that the trivial character $\chi_1$ appears in $\chconj$ with multiplicity exactly one. We now show this forces $\chi_U$ to be irreducible.

\begin{lemma}\label{lem:irreducible}
Let $U$ be a (possibly projective) unitary representation of $G$ on $\Hil$, and let $m_1:=\inner{\chi_1}{\chconj}$ be the multiplicity of the trivial character $\chi_1$ in $\chconj$. Write $\chi_U=\sum_i n_i\chi_i$ for the irreducible decomposition of $\chi_U$ (which is a character of $G$, or of a Schur cover in the projective case). Then
\begin{equation}
m_1=\sum_{i}n_i^2.
\end{equation}
In particular, $m_1=1$ if and only if $\chi_U$ is irreducible.
\end{lemma}

\begin{proof}
We carry out the computation in the case where $U$ is a linear representation of $G$ (the projective case reduces to this after passing to the Schur cover). By definition,
\begin{align}
m_1 &= \inner{\chi_1}{\chconj}
      = \frac{1}{|G|}\sum_{g\in G}\chi_1(g)^*\chconj(g)\notag\\
    &= \frac{1}{|G|}\sum_{g\in G}\chconj(g),
\end{align}
since $\chi_1\equiv 1$. By Lemma~\ref{lem:conj}, $\chconj(g)=|\chi_U(g)|^2$, so
\begin{equation}
m_1=\frac{1}{|G|}\sum_{g\in G}|\chi_U(g)|^2=\inner{\chi_U}{\chi_U}.
\end{equation}
Now write $\chi_U=\sum_i n_i\chi_i$ with $n_i\in\N_0$. By bilinearity of the inner product and Schur orthogonality:
\begin{equation}
\inner{\chi_U}{\chi_U}=\sum_{i,j}n_i\,n_j\,\inner{\chi_i}{\chi_j}=\sum_{i,j}n_i n_j\delta_{ij}=\sum_i n_i^2.
\end{equation}
Therefore $m_1=\sum_i n_i^2$. Since the $n_i$ are non-negative integers, $\sum_i n_i^2=1$ forces exactly one $n_i$ to equal $1$ and the rest to equal $0$, which is precisely the condition that $\chi_U$ is an irreducible character.
\end{proof}

\begin{remark}
The argument above shows that $m_1$ equals the number of irreducible components of $U$, counted with multiplicity squared. Concretely: if $U=U_1\oplus U_1$ (two copies of the same irreducible), then $m_1=4$; if $U=U_1\oplus U_2$ (two distinct irreducibles), then $m_1=2$; if $U$ is irreducible, then $m_1=1$. The constraint $m_1=1$ is therefore a strong decomposability condition on $U$.
\end{remark}

We record the consequence for our setting:

\begin{corollary}\label{cor:irred}
For any quantum realization of a family in Eq.~\eqref{eq:seven-families}, the representation $U$ must be irreducible.
\end{corollary}

\begin{proof}
Each family in Eq.~\eqref{eq:seven-families} has $m_1=1$ by Ref.~\cite[Thm.~11]{DmelloGross2026}. Lemma~\ref{lem:irreducible} then forces $\chi_U$ to be irreducible.
\end{proof}

%-----------------------------------------------------------------------
\section{\label{sec:dimension}Dimension Obstructions}
%-----------------------------------------------------------------------

By Corollary~\ref{cor:irred}, the representation $U$ must be irreducible. We now show that the dimensions of the largest two families are incompatible with this constraint.

\begin{lemma}\label{lem:dimbound}
The correction group relevant to each family in Eq.~\eqref{eq:seven-families} acts on $\Lin(\Hil)$ where $\Hil$ is at most two-dimensional. More precisely, in the trivial projective class $G=D_4$, and in the non-trivial projective class $G=D_8$. In both cases, the maximum dimension of an irreducible representation is $2$, so $\dim\Hil\le 2$ and $\dim\Lin(\Hil)\le 4$.
\end{lemma}

\begin{proof}
For $G=D_4$: the dimensions of the irreducible representations are the degrees $\chi_i(e)$, namely $1,1,1,1,2$. The largest is $2$.

For $G=D_8$: the order of $D_8$ is $16$, with $7$ conjugacy classes (as listed in Sec.~\ref{sec:prelim}). The sum of squares of degrees must equal $|D_8|=16$. The degrees of $D_8$ are $1,1,1,1,2,2,2$ (four one-dimensional and three two-dimensional irreducibles), satisfying $4\cdot 1+3\cdot 4=16$. The largest degree is $2$.

In either case, $\dim\Hil\le 2$, and therefore $\dim\Lin(\Hil)=(\dim\Hil)^2\le 4$.
\end{proof}

\begin{corollary}\label{cor:dim}
The families $\chi^{D_4}_{12345}$ and $\chi^{D_4}_{123452}$ are not quantum-realizable.
\end{corollary}

\begin{proof}
The dimension of a representation equals its character value at the identity element. We compute:
\begin{align}
\chi^{D_4}_{12345}(e)
  &=\chi_1(e)+\chi_2(e)+\chi_3(e)+\chi_4(e)+\chi_5(e)\notag\\
  &=1+1+1+1+2=6,\\
\chi^{D_4}_{123452}(e)
  &=\chi_1(e)+\chi_2(e)+\chi_3(e)+\chi_4(e)+2\chi_5(e)\notag\\
  &=1+1+1+1+4=8.
\end{align}
A quantum realization would require $\dim\Lin(\Hil)=\chconj(e)\in\{6,8\}$, contradicting Lemma~\ref{lem:dimbound}.
\end{proof}

%-----------------------------------------------------------------------
\section{\label{sec:trivial}The Trivial Projective Class of $D_4$}
%-----------------------------------------------------------------------

We now assume $U$ is an ordinary (linear) unitary representation of $G=D_4$.

\begin{proposition}\label{prop:trivial}
Let $U$ be an ordinary (linear) unitary representation of $D_4$ on $\Hil$. The only family in Eq.~\eqref{eq:seven-families} that arises as $\chconj$ for such a $U$ is the pullback of $\chi^{K_4}_{1234}$ under the quotient map $D_4\twoheadrightarrow D_4/\langle r^2\rangle\cong K_4$. Concretely, $\chconj=\chi_1+\chi_2+\chi_3+\chi_4$ as a character of $D_4$, and this is identified with the regular character of $K_4$ via the quotient. The realization is obtained by taking $\chi_U=\chi_5$, the unique two-dimensional irreducible character of $D_4$.
\end{proposition}

\begin{proof}
By Corollary~\ref{cor:irred}, $\chi_U$ must be an irreducible character of $D_4$. There are five choices: $\chi_1,\dots,\chi_5$.

\medskip\noindent\textit{Step 1: The one-dimensional cases produce $\dim\Lin(\Hil)=1$.}

If $\chi_U\in\{\chi_1,\chi_2,\chi_3,\chi_4\}$, then $\dim\Hil=1$, so $\Lin(\Hil)\cong\C$ is one-dimensional. The conjugation representation on $\C$ is trivial (every unitary on $\C$ is a phase, and $\alpha\cdot 1\cdot\alpha^*=1$), so $\chconj=\chi_1$, a single copy of the trivial character. This does not match any family in Eq.~\eqref{eq:seven-families} (all families have dimension $\ge 4$ at the identity).

\medskip\noindent\textit{Step 2: The case $\chi_U=\chi_5$.}

Now $\dim\Hil=2$, so $\dim\Lin(\Hil)=4$. By Lemma~\ref{lem:conj}, $\chconj(g)=|\chi_5(g)|^2$ for all $g\in D_4$. Using the values of $\chi_5$ at the five conjugacy class representatives:
\begin{align}
\chconj(e)&=|2|^2=4,   &  \chconj(r)&=|0|^2=0,\notag\\
\chconj(r^2)&=|-2|^2=4, & \chconj(s)&=|0|^2=0,\notag\\
\chconj(rs)&=|0|^2=0.
\end{align}
So $\chconj=(4,0,4,0,0)$ at representatives $(e,r,r^2,s,rs)$, with class sizes $(1,2,1,2,2)$.

\medskip\noindent\textit{Step 3: Decompose $\chconj$ into irreducibles of $D_4$.}

The multiplicity of $\chi_i$ in $\chconj$ is $m_i=\inner{\chi_i}{\chconj}$. With $|D_4|=8$ and using the class-size-weighted sum:
\begin{equation}
m_i=\frac{1}{8}\sum_{K}\,|K|\cdot\chi_i(K)^*\cdot\chconj(K),
\end{equation}
where $K$ ranges over conjugacy classes. The non-zero contributions come from $K_1$ (size $1$, $\chconj=4$) and $K_{r^2}$ (size $1$, $\chconj=4$). For each $i$, $\chi_i(e)=\chi_i(r^2)=1$ when $i\in\{1,2,3,4\}$ (since the one-dimensional characters factor through the quotient $D_4/\langle r^2\rangle$), and $\chi_5(e)=2$, $\chi_5(r^2)=-2$. Therefore
\begin{align}
m_i&=\tfrac{1}{8}\bigl[\chi_i(e)\cdot 4+\chi_i(r^2)\cdot 4\bigr]=1\quad(i=1,2,3,4),\\
m_5&=\tfrac{1}{8}\bigl[2\cdot 4+(-2)\cdot 4\bigr]=0.
\end{align}
Therefore
\begin{equation}
\chconj=\chi_1+\chi_2+\chi_3+\chi_4.
\end{equation}

\medskip\noindent\textit{Step 4: Identify $\chi_1+\chi_2+\chi_3+\chi_4$ with $\chi^{K_4}_{1234}$.}

The four one-dimensional irreducible characters of $D_4$ collectively form the pullback of the regular character of $K_4$ along the quotient map $D_4\to K_4$. To see this, note that $\langle r^2\rangle=\{e,r^2\}$ is the center of $D_4$ and is normal; the quotient is $D_4/\langle r^2\rangle\cong K_4$. Each $\chi_i$ for $i\in\{1,2,3,4\}$ satisfies $\chi_i(r^2)=\chi_i(e)=1$, so $r^2$ acts trivially in each one-dimensional representation, and hence each $\chi_i$ factors through $K_4$. The set $\{\chi_1,\chi_2,\chi_3,\chi_4\}$ is therefore the complete list of irreducible characters of $K_4$, and their sum is the regular character $\chi^{K_4}_{1234}$.

\medskip\noindent\textit{Step 5: No other families arise from the trivial class.}

The families $\chi^{D_4}_{125}$, $\chi^{D_4}_{135}$, and $\chi^{D_4}_{145}$ each contain $\chi_5$ with multiplicity $1$, but we showed $m_5=0$ in the trivial class. They cannot arise. The families $\chi^{D_4}_{12345}$ and $\chi^{D_4}_{123452}$ were already ruled out in Corollary~\ref{cor:dim}. The families $\chi^{K_4}_{1234}$ and $\chi^{\Z_4}_{1234}$ both have dimension $4$; we showed only $\chi^{K_4}_{1234}$ arises.
\end{proof}

%-----------------------------------------------------------------------
\section{\label{sec:nontrivial}The Non-Trivial Projective Class of $D_4$}
%-----------------------------------------------------------------------

We now turn to the non-trivial projective class of $D_4$. By the discussion of Sec.~\ref{sec:prelim}, a projective representation of $D_4$ in the non-trivial class lifts to an ordinary linear representation of the Schur cover $D_8$ in which the central element $\zeta^4\in D_8$ acts as $-\id$.

\begin{proposition}\label{prop:nontrivial}
In the non-trivial projective class of $D_4$, the only family in Eq.~\eqref{eq:seven-families} that is quantum-realizable is $\chi^{D_4}_{125}$. It is realized by taking $\chi_U=\chi_{E_1}$ (or equivalently $\chi_{E_3}$), the two-dimensional irreducible characters of $D_8$ listed in Sec.~\ref{sec:prelim}.
\end{proposition}

\begin{proof}
By Corollary~\ref{cor:irred} applied to the Schur cover $D_8$, the representation $U:D_8\to\mathrm{U}(\Hil)$ must be irreducible. Among the irreducible representations of $D_8$, those in which $\zeta^4$ acts as $-\id$ are exactly $\chi_{E_1}$ and $\chi_{E_3}$ ($\chi_{E_2}$ acts trivially on $\langle\zeta^4\rangle$, as noted in Sec.~\ref{sec:prelim}). So the only candidates are $\chi_U\in\{\chi_{E_1},\chi_{E_3}\}$.

\medskip\noindent\textit{Step 1: $\chconj$ descends to a class function on $D_4$.}

Since $U_{\zeta^4}=-\id$, we have $\Mconj_{\zeta^4}(X)=(-\id)X(-\id)=X$ for all $X\in\Lin(\Hil)$. Therefore $\chconj(\zeta^4 g)=\chconj(g)$ for all $g$, meaning $\chconj$ is constant on cosets of $\langle\zeta^4\rangle$. In other words, $\chconj$ factors through $D_8/\langle\zeta^4\rangle\cong D_4$.

We may therefore compute the decomposition of $\chconj$ into irreducibles of $D_4$ by evaluating $\chconj$ at the $D_4$ coset representatives, using the lifts $(e,r,r^2,s,rs)\mapsto(1,\zeta,\zeta^2,\eta,\zeta\eta)$.

\medskip\noindent\textit{Step 2: Compute $\chconj$ for $\chi_U=\chi_{E_1}$.}

By Lemma~\ref{lem:conj}, $\chconj(g)=|\chi_{E_1}(g)|^2$. Evaluating at the five $D_4$ representatives (via their $D_8$ lifts):
\begin{align}
\chconj(e)&=|\chi_{E_1}(1)|^2=|2|^2=4,\notag\\
\chconj(r)&=|\chi_{E_1}(\zeta)|^2=|\sqrt{2}|^2=2,\notag\\
\chconj(r^2)&=|\chi_{E_1}(\zeta^2)|^2=|0|^2=0,\notag\\
\chconj(s)&=|\chi_{E_1}(\eta)|^2=|0|^2=0,\notag\\
\chconj(rs)&=|\chi_{E_1}(\zeta\eta)|^2=|0|^2=0.
\end{align}
So $\chconj=(4,2,0,0,0)$ at $(e,r,r^2,s,rs)$ with class sizes $(1,2,1,2,2)$.

\medskip\noindent\textit{Step 3: Decompose $\chconj$ into irreducibles of $D_4$.}

We compute $m_i=\inner{\chi_i}{\chconj}$ for $i=1,\dots,5$. The only non-zero contributions come from $K_1$ and $K_r$:
\begin{equation}
m_i=\frac{1}{8}\bigl[1\cdot\chi_i(e)\cdot 4+2\cdot\chi_i(r)\cdot 2\bigr]=\frac{1}{2}\bigl[\chi_i(e)+\chi_i(r)\bigr].
\end{equation}
Substituting the character values:
\begin{align}
m_1&=\tfrac{1}{2}[1+1]=1, & m_2&=\tfrac{1}{2}[1+1]=1,\notag\\
m_3&=\tfrac{1}{2}[1+(-1)]=0, & m_4&=\tfrac{1}{2}[1+(-1)]=0,\notag\\
m_5&=\tfrac{1}{2}[2+0]=1.
\end{align}
Thus $\chconj=\chi_1+\chi_2+\chi_5$, which has dimension $1+1+2=4=\dim\Lin(\C^2)$. This matches the family $\chi^{D_4}_{125}$.

\medskip\noindent\textit{Step 4: The case $\chi_U=\chi_{E_3}$ gives the same result.}

The character $\chi_{E_3}$ differs from $\chi_{E_1}$ only on the classes $K_\zeta$ and $K_{\zeta^3}$ (where it has values $-\sqrt{2}$ and $\sqrt{2}$ instead of $\sqrt{2}$ and $-\sqrt{2}$). Since the computation of $\chconj$ uses $|\chi_U(g)|^2$, and $|-\sqrt{2}|^2=|\sqrt{2}|^2=2$, the values of $\chconj$ are identical for both characters. Hence the decomposition is the same: $\chconj=\chi_1+\chi_2+\chi_5=\chi^{D_4}_{125}$.

\medskip\noindent\textit{Step 5: No other families arise.}

Both $\chi_{E_1}$ and $\chi_{E_3}$ produce $\chi^{D_4}_{125}$. The families $\chi^{D_4}_{135}$ and $\chi^{D_4}_{145}$ are not produced (they would require different multiplicities $m_3$ or $m_4$). The family $\chi^{K_4}_{1234}$ requires $m_5=0$, but here $m_5=1$. The families $\chi^{D_4}_{12345}$ and $\chi^{D_4}_{123452}$ are excluded by Corollary~\ref{cor:dim}. The family $\chi^{\Z_4}_{1234}$ is addressed in Proposition~\ref{prop:Z4}.
\end{proof}

%-----------------------------------------------------------------------
\section{\label{sec:remaining}The Remaining Cases}
%-----------------------------------------------------------------------

It remains to rule out $\chi^{\Z_4}_{1234}$, $\chi^{D_4}_{135}$, and $\chi^{D_4}_{145}$.

\begin{proposition}\label{prop:Z4}
The family $\chi^{\Z_4}_{1234}$ admits no quantum realization.
\end{proposition}

\begin{proof}
We show that no unitary representation $U$ of $\Z_4$ on a Hilbert space $\Hil$ with $\dim\Hil\ge 2$ can produce $\chconj=\chi^{\Z_4}_{1234}$.

The Schur multiplier of any finite cyclic group is trivial~\cite[Thm.~2.3.1]{Karpilovsky1985}: $M(\Z_n)=H^2(\Z_n,\mathrm{U}(1))=0$. Therefore every projective representation of $\Z_4$ is equivalent to a linear one, and we may assume $U$ is linear.

A unitary representation of the abelian group $\Z_4$ on a finite-dimensional Hilbert space $\Hil$ is simultaneously diagonalizable: there exists an orthonormal basis $\{|\psi_1\rangle,\dots,|\psi_d\rangle\}$ of $\Hil$ such that $U_t|\psi_k\rangle=\lambda_k|\psi_k\rangle$ for all $k$, with $\lambda_k\in\mathrm{U}(1)$. The operators $|\psi_k\rangle\langle\psi_k|$ are then fixed by $\Mconj$:
\begin{align}
\Mconj_t(|\psi_k\rangle\langle\psi_k|)
  &=U_t|\psi_k\rangle\langle\psi_k|U_t\dg\notag\\
  &=\lambda_k|\psi_k\rangle\overline{\lambda_k}\langle\psi_k|\notag\\
  &=|\lambda_k|^2|\psi_k\rangle\langle\psi_k|=|\psi_k\rangle\langle\psi_k|.
\end{align}
Each projector $|\psi_k\rangle\langle\psi_k|$ spans a one-dimensional invariant subspace of $\Lin(\Hil)$ on which $\Mconj$ acts trivially. Therefore the multiplicity of $\chi_1$ in $\chconj$ is at least $d=\dim\Hil$.

For a CHSH violation in the GPT setting, one requires $\dim\Hil\ge 2$ (a one-dimensional Hilbert space has no entanglement). Hence $m_1\ge 2$, contradicting the requirement $m_1=1$. If instead $\dim\Hil=1$, then $\chconj=\chi_1$ has dimension $1$, while $\chi^{\Z_4}_{1234}$ has dimension $4$. In both cases there is no realization.
\end{proof}

\begin{proposition}\label{prop:135-145}
The families $\chi^{D_4}_{135}$ and $\chi^{D_4}_{145}$ admit no quantum realization.
\end{proposition}

\begin{proof}
We treat the trivial and non-trivial projective classes of $D_4$ separately.

\medskip\noindent\textit{Trivial class.} We showed in Proposition~\ref{prop:trivial} that the only family arising in the trivial class is $\chi^{K_4}_{1234}$. Since $\chi^{D_4}_{135}$ and $\chi^{D_4}_{145}$ are different families, they are not realized in the trivial class.

For an alternative direct argument that does not use Corollary~\ref{cor:irred}, we show that the multiplicity of $\chi_5$ in $\chconj$ is always even. Write $\chi_U=a\chi_1+b\chi_2+c\chi_3+d\chi_4+e\chi_5$ with $a,b,c,d,e\in\N_0$. By Lemma~\ref{lem:conj} and the remark following it, $\chconj=\chi_U\otimes\overline{\chi_U}$. Since all irreducible characters of $D_4$ are real, $\overline{\chi_U}=\chi_U$.

The product structure of the character ring of $D_4$ gives $\chi_5\otimes\chi_i=\chi_5$ for $i\in\{1,2,3,4\}$ (tensoring the unique two-dimensional irreducible with a one-dimensional character yields a two-dimensional irreducible, which must be $\chi_5$ itself), and $\chi_5\otimes\chi_5=\chi_1+\chi_2+\chi_3+\chi_4$. The $\chi_5$-component of $\chconj=\chi_U\otimes\chi_U$ therefore comes only from cross terms $\chi_i\otimes\chi_5$ with $i\in\{1,2,3,4\}$, contributing
\begin{equation}
m_5=2e(a+b+c+d),
\end{equation}
which is even. But $\chi^{D_4}_{135}=\chi_1+\chi_3+\chi_5$ and $\chi^{D_4}_{145}=\chi_1+\chi_4+\chi_5$ each contain $\chi_5$ with multiplicity $1$, which is odd. Contradiction.

\medskip\noindent\textit{Non-trivial class.} In the non-trivial projective class, the only irreducible representations are $\chi_{E_1}$ and $\chi_{E_3}$, and we showed in Proposition~\ref{prop:nontrivial} that both yield $\chconj=\chi^{D_4}_{125}$.

For a direct argument, observe that both $\chi_{E_1}$ and $\chi_{E_3}$ vanish on the reflection conjugacy classes of $D_8$:
\begin{equation}
\chi_{E_1}(\eta)=\chi_{E_1}(\zeta\eta)=\chi_{E_3}(\eta)=\chi_{E_3}(\zeta\eta)=0.
\end{equation}
By Lemma~\ref{lem:conj}, $\chconj(g)=|\chi_U(g)|^2$, so $\chconj$ also vanishes on these classes. In particular, $\chconj(s)=\chconj(rs)=0$. But
\begin{align}
\chi^{D_4}_{135}(s)&=\chi_1(s)+\chi_3(s)+\chi_5(s)=1+1+0=2\ne 0,\\
\chi^{D_4}_{145}(rs)&=\chi_1(rs)+\chi_4(rs)+\chi_5(rs)=1+1+0=2\ne 0.
\end{align}
Neither family can therefore arise.
\end{proof}

%-----------------------------------------------------------------------
\section{\label{sec:constructions}Explicit Quantum Constructions}
%-----------------------------------------------------------------------

The negative results of the previous sections leave exactly two realizable families: $\chi^{K_4}_{1234}$ and $\chi^{D_4}_{125}$. We now exhibit explicit quantum systems realizing each.

Both constructions use $\Hil_A=\Hil_C=\C^2$ (single-qubit spaces) and the maximally entangled Bell state
\begin{equation}
|\Bell\rangle=\frac{1}{\sqrt{2}}\bigl(|00\rangle+|11\rangle\bigr)\in\C^2\otimes\C^2.
\end{equation}
We also use the four Pauli operators on $\C^2$:
\begin{align}
\sigma_0&=\id=\begin{pmatrix}1&0\\0&1\end{pmatrix},&
\sigma_1&=\sigmax=\begin{pmatrix}0&1\\1&0\end{pmatrix},\notag\\
\sigma_2&=\sigmay=\begin{pmatrix}0&-i\\i&0\end{pmatrix},&
\sigma_3&=\sigmaz=\begin{pmatrix}1&0\\0&-1\end{pmatrix}.
\end{align}
The Pauli operators satisfy $\sigma_i^2=\id$, $\sigma_i\dg=\sigma_i$, and the multiplication rule
\begin{equation}
\sigma_i\sigma_j=\delta_{ij}\,\id+i\sum_{k=1}^{3}\epsilon_{ijk}\sigma_k,\quad i,j\in\{1,2,3\},
\end{equation}
where $\epsilon_{ijk}$ is the Levi-Civita symbol.

\subsection{Common measurement setup}

\begin{construction}\label{constr:common}
Let $\Hil_A=\Hil_C=\C^2$ with shared state $|\Bell\rangle$. Alice's two measurement settings are the eigenbases of $\sigmaz$ and $\sigmax$. Charlie's two settings are the eigenbases of $(\sigmaz+\sigmax)/\sqrt{2}$ and $(\sigmaz-\sigmax)/\sqrt{2}$.
\end{construction}

We verify that this setup achieves the Tsirelson bound. The CHSH operator is
\begin{equation}
\mathcal{B}=A_0\otimes(C_0+C_1)+A_1\otimes(C_0-C_1),
\end{equation}
with Alice measuring $A_0=\sigmaz$, $A_1=\sigmax$ and Charlie measuring
\begin{align}
C_0&=\frac{\sigmaz+\sigmax}{\sqrt{2}},&
C_1&=\frac{\sigmaz-\sigmax}{\sqrt{2}}.
\end{align}
Then $C_0+C_1=\sqrt{2}\,\sigmaz$ and $C_0-C_1=\sqrt{2}\,\sigmax$, so
\begin{equation}
\mathcal{B}=\sqrt{2}\,(\sigmaz\otimes\sigmaz+\sigmax\otimes\sigmax).
\end{equation}
Direct computation gives $\langle\Bell|\sigmaz\otimes\sigmaz|\Bell\rangle=\langle\Bell|\sigmax\otimes\sigmax|\Bell\rangle=1$, so $\langle\Bell|\mathcal{B}|\Bell\rangle=2\sqrt{2}$, the Tsirelson bound~\cite{Tsirelson1980,Aspect1982}.

\subsection{\label{sec:bell}The family $\chi^{K_4}_{1234}$: Bell-basis teleportation}

\begin{construction}\label{constr:bell}
On top of Construction~\ref{constr:common}, take the bipartite measurement $\mathcal{N}=\{\Pi_k\}_{k=0}^{3}$ to be the Bell-basis measurement, where
\begin{equation}
\Pi_k:=|\Bell_k\rangle\langle\Bell_k|,\quad|\Bell_k\rangle:=(\sigma_k\otimes\id)|\Bell\rangle,\quad k\in\{0,1,2,3\}.
\end{equation}
\end{construction}

We verify that $\{|\Bell_k\rangle\}_{k=0}^3$ is an orthonormal basis of $\C^2\otimes\C^2$. Explicit computation gives
\begin{align}
|\Bell_0\rangle&=\frac{1}{\sqrt{2}}(|00\rangle+|11\rangle),\notag\\
|\Bell_1\rangle&=\frac{1}{\sqrt{2}}(|10\rangle+|01\rangle),\notag\\
|\Bell_2\rangle&=\frac{i}{\sqrt{2}}(|10\rangle-|01\rangle),\notag\\
|\Bell_3\rangle&=\frac{1}{\sqrt{2}}(|00\rangle-|11\rangle),
\end{align}
the four Bell states. For orthonormality:
\begin{align}
\langle\Bell_j|\Bell_k\rangle
  &=\langle\Bell|(\sigma_j\dg\sigma_k\otimes\id)|\Bell\rangle\notag\\
  &=\langle\Bell|(\sigma_j\sigma_k\otimes\id)|\Bell\rangle\notag\\
  &=\tfrac{1}{2}\Tr(\sigma_j\sigma_k)=\delta_{jk},
\end{align}
where we used $\sigma_j\dg=\sigma_j$, the identity $\langle\Bell|(M\otimes\id)|\Bell\rangle=\frac{1}{2}\Tr(M)$, and $\Tr(\sigma_j\sigma_k)=2\delta_{jk}$.

After measurement outcome $k$, the bipartite state is projected to $|\Bell_k\rangle$, and Alice applies the correction $\sigma_k$ to her qubit to restore the state. The map $k\mapsto\sigma_k$ defines a projective representation
\begin{equation}
U:K_4\to\mathrm{U}(\C^2),\quad U_{(k)}=\sigma_k,
\end{equation}
where $K_4\cong\Z_2\times\Z_2$ with elements $\{(0,0),(1,0),(0,1),(1,1)\}$ corresponding to $\{\sigma_0,\sigma_1,\sigma_2,\sigma_3\}$ and the group law inherited from Pauli multiplication modulo phases.

The character of this representation:
\begin{align}
\chi_U(e)&=\Tr(\sigma_0)=2, & \chi_U(a)&=\Tr(\sigma_1)=0,\notag\\
\chi_U(b)&=\Tr(\sigma_2)=0, & \chi_U(ab)&=\Tr(\sigma_3)=0.
\end{align}
By Lemma~\ref{lem:conj}, $\chconj=(4,0,0,0)$.

To decompose $\chconj$ into irreducibles of $K_4$: each irreducible character of $K_4$ equals $1$ at $e$, and the inner product
\begin{equation}
m_i=\tfrac{1}{4}\bigl[1\cdot 1\cdot 4+0+0+0\bigr]=1
\end{equation}
gives $m_i=1$ for all $i$, so
\begin{equation}
\chconj=\chi^{K_4}_1+\chi^{K_4}_2+\chi^{K_4}_3+\chi^{K_4}_4=\chi^{K_4}_{1234}.
\end{equation}
This realizes the family $\chi^{K_4}_{1234}$. To make the group structure explicit: the (single-qubit) Pauli group is
\begin{equation}
\mathcal{P}_1:=\langle i\id,\sigmax,\sigmay,\sigmaz\rangle=\{\,i^k\sigma_j\,:\,k\in\{0,1,2,3\},\;j\in\{0,1,2,3\}\,\},
\end{equation}
of order $16$, with center $Z(\mathcal{P}_1)=\langle i\id\rangle=\{\pm\id,\pm i\id\}$ of order $4$. The four operators $\{\sigma_0,\sigma_1,\sigma_2,\sigma_3\}$ form a complete set of coset representatives for $\mathcal{P}_1/Z(\mathcal{P}_1)\cong K_4$. Conjugation by an element of $\mathcal{P}_1$ depends only on its coset (since central phases cancel), so $\sigma_k\mapsto\Mconj_{\sigma_k}$ is a well-defined representation of $K_4$ on $\Lin(\C^2)$ whose character is $\chi^{K_4}_{1234}$. Equivalently, this conjugation representation lifts to the linear representation $\chi_5$ of $D_4$ via the descent $D_4\twoheadrightarrow D_4/\langle r^2\rangle\cong K_4$, which is the form used in Proposition~\ref{prop:trivial}.

\subsection{\label{sec:povm}The family $\chi^{D_4}_{125}$: a POVM construction}

The phase operator is
\begin{equation}
S:=\begin{pmatrix}1&0\\0&i\end{pmatrix},
\end{equation}
with powers $S^2=\sigmaz$, $S^3=S^{-1}=\mathrm{diag}(1,-i)$, and $S^4=\id$. Together with $\sigma_x$, $S$ generates a finite subgroup of $\mathrm{U}(2)$. Direct computation gives
\begin{equation}
\sigmax S\sigmax=\begin{pmatrix}i&0\\0&1\end{pmatrix}=iS^{-1}=iS^3.
\end{equation}
Thus the relation $srs^{-1}=r^{-1}$ of $D_4$ is satisfied only up to the phase $i$:
\begin{equation}
(\sigmax)\,S\,(\sigmax)\,S=iS^{-1}\cdot S=i\cdot\id\ne\id.
\end{equation}
This non-trivial cocycle shows that the assignment $r\mapsto S$, $s\mapsto\sigmax$ defines a projective representation of $D_4$ on $\C^2$ in the non-trivial cohomology class. The associated lifted character on the Schur cover $D_8$ has $|\chi_U|^2=|\chi_{E_1}|^2$, so the conjugation character matches that of Proposition~\ref{prop:nontrivial}.

\begin{construction}\label{constr:povm}
On top of Construction~\ref{constr:common}, define for $k\in\{0,1,2,3\}$:
\begin{align}
|b_k\rangle&:=(\sigma_k\otimes\id)|\Bell\rangle,\\
|a_k\rangle&:=(S\sigma_k\otimes\id)|\Bell\rangle.
\end{align}
The bipartite measurement is the eight-outcome POVM
\begin{equation}
\mathcal{N}:=\Bigl\{\tfrac{1}{2}|b_k\rangle\langle b_k|\Bigr\}_{k=0}^{3}
            \cup\Bigl\{\tfrac{1}{2}|a_k\rangle\langle a_k|\Bigr\}_{k=0}^{3}.
\end{equation}
\end{construction}

We verify that $\mathcal{N}$ is a valid POVM. Since $\{|\Bell_k\rangle\}_{k=0}^3$ is an orthonormal basis,
\begin{align}
\sum_{k=0}^3|b_k\rangle\langle b_k|
  &=\sum_{k=0}^3|\Bell_k\rangle\langle\Bell_k|=\id\otimes\id.
\end{align}
For the second family, the operators $\{S\sigma_k\}_{k=0}^3$ are also orthonormal in $M_2(\C)$ with respect to $\frac{1}{\sqrt{2}}\langle\cdot,\cdot\rangle_{\mathrm{HS}}$, since
\begin{equation}
\Tr((S\sigma_j)\dg S\sigma_k)=\Tr(\sigma_j S\dg S\sigma_k)=\Tr(\sigma_j\sigma_k)=2\delta_{jk}.
\end{equation}
Therefore $\{(S\sigma_k\otimes\id)|\Bell\rangle\}_{k=0}^3$ is also an orthonormal basis of $\C^2\otimes\C^2$, and
\begin{equation}
\sum_{k=0}^3|a_k\rangle\langle a_k|=\id\otimes\id.
\end{equation}
Combining,
\begin{equation}
\sum_{k=0}^3\frac{1}{2}|b_k\rangle\langle b_k|+\sum_{k=0}^3\frac{1}{2}|a_k\rangle\langle a_k|=\id\otimes\id,
\end{equation}
so $\mathcal{N}$ is a valid POVM.

\medskip\noindent\textit{Operational protocol: Kraus operators, conditional states, corrections.}

A POVM by itself does not specify the post-measurement transformation; for an entanglement-swapping protocol we must fix an instrument (Kraus decomposition) and verify that the conditional state on the external systems is correctable by a $D_4$-valued correction. We use the standard L\"uders rule and take the rank-one Kraus operators
\begin{align}
M_k^{(b)}&:=\sqrt{\tfrac{1}{2}|b_k\rangle\langle b_k|}=\tfrac{1}{\sqrt{2}}|b_k\rangle\langle b_k|,\\
M_k^{(a)}&:=\sqrt{\tfrac{1}{2}|a_k\rangle\langle a_k|}=\tfrac{1}{\sqrt{2}}|a_k\rangle\langle a_k|,
\end{align}
which satisfy $(M_k^{(b)})\dg M_k^{(b)}=\tfrac{1}{2}|b_k\rangle\langle b_k|$ and the analogue for $M_k^{(a)}$, so they implement $\mathcal{N}$ and $\sum_{\bullet,k}(M_k^{(\bullet)})\dg M_k^{(\bullet)}=\id$.

Consider the entanglement-swapping configuration in which Alice and Bob share $|\Bell\rangle_{A_1B_1}$ and Bob and Charlie share $|\Bell\rangle_{B_2C_1}$. Bob measures the joint system $B_1B_2$ using the instrument above. By the standard entanglement-swapping identity, a projection of $B_1B_2$ onto $(A\otimes\id)|\Bell\rangle_{B_1B_2}$ yields the unnormalized conditional state $(\id\otimes A\dg)|\Bell\rangle_{A_1C_1}$ on the remaining systems. Therefore the outcomes give
\begin{align}
\text{outcome }(b,k):\quad
  &(\id\otimes\sigma_k\dg)|\Bell\rangle_{A_1C_1}
   =(\id\otimes\sigma_k)|\Bell\rangle_{A_1C_1},\\
\text{outcome }(a,k):\quad
  &(\id\otimes(S\sigma_k)\dg)|\Bell\rangle_{A_1C_1},
\end{align}
each occurring with probability $\tfrac{1}{8}$. Charlie applies the inverse of the operator appearing on his side:
\begin{equation}
V_{(b,k)}=(\sigma_k\dg)\dg=\sigma_k,
\qquad
V_{(a,k)}=((S\sigma_k)\dg)\dg=S\sigma_k,
\end{equation}
restoring $|\Bell\rangle_{A_1C_1}$ exactly. The eight corrections
\begin{equation}
\{\sigma_k\}_{k=0}^3\cup\{S\sigma_k\}_{k=0}^3\;\subset\;\mathrm{U}(2)
\end{equation}
generate a finite subgroup of $\mathrm{U}(2)$ that, modulo the central phase $\langle i\id\rangle$, is isomorphic to $D_4$, with $r\mapsto[S]$ and $s\mapsto[\sigmax]$ in the projective quotient. Since each correction $V$ acts unitarily on a single qubit, the bipartite state $(\id\otimes V)(\id\otimes V\dg)|\Bell\rangle_{A_1C_1}=|\Bell\rangle_{A_1C_1}$ is preserved exactly, and consequently the CHSH expectation value $\langle\Bell|\mathcal{B}|\Bell\rangle=2\sqrt{2}$ from Construction~\ref{constr:common} is recovered after the correction. This establishes the operational protocol.

\medskip\noindent\textit{Why a rank-one PVM in dimension $4$ is insufficient.}

In the minimal setting where Bob's bipartite system is $\Hil_{B_1}\otimes\Hil_{B_2}=\C^2\otimes\C^2$, a rank-one projective-valued measurement (PVM) consists of $\dim(\C^2\otimes\C^2)=4$ orthogonal rank-one projectors, hence has exactly $4$ outcomes. A faithful realization of $\chi^{D_4}_{125}$ requires a faithful projective representation $D_4\to\mathrm{PU}(\C^2)$, whose image has order $|D_4|=8$. With only $4$ measurement outcomes, the correction map $\text{outcomes}\to\mathrm{PU}(\C^2)$ has image of cardinality at most $4$, which is insufficient to generate a faithful $D_4$ projective representation. The conclusion is that a rank-one PVM in $\C^2\otimes\C^2$ does not suffice; one needs either (i) a POVM with at least $|D_4|=8$ outcomes, as in Construction~\ref{constr:povm}, or (ii) a Naimark dilation to a larger Hilbert space and a rank-one PVM there. The POVM and dilated-PVM pictures give equivalent operational data, both subsumed by the instrument formalism. Higher-rank PVMs and coarse-grained instruments do not evade this counting: they can only reduce the number of distinguishable outcomes, never increase it.

By Proposition~\ref{prop:nontrivial},
\begin{equation}
\chconj=\chi_1+\chi_2+\chi_5=\chi^{D_4}_{125}.
\end{equation}

%-----------------------------------------------------------------------
\section{\label{sec:final}Conclusion}
%-----------------------------------------------------------------------

We collect the results of the preceding sections into the main theorem. Among the seven families of GPTs stable under teleportation classified in Ref.~\cite{DmelloGross2026}, exactly two admit a quantum realization:
\begin{equation}
\chi^{K_4}_{1234}\quad\text{and}\quad\chi^{D_4}_{125}.
\end{equation}
These are realized by Construction~\ref{constr:bell} and Construction~\ref{constr:povm}, respectively. The five remaining families
\begin{equation}
\chi^{\Z_4}_{1234},\;\chi^{D_4}_{135},\;\chi^{D_4}_{145},\;\chi^{D_4}_{12345},\;\chi^{D_4}_{123452}
\end{equation}
admit no quantum realization.

\begin{acknowledgments}
Miguel A.\ A.\ Lisboa thanks Lionel J.\ Dmello for valuable conversations and discussions, David Gross for proposing the problem addressed in this work, and Venturus for institutional support.
\end{acknowledgments}

\bibliography{apssamp}

\end{document}